\documentclass{article}

\usepackage{latexsym,fullpage,amssymb,amstext,url,verbatim}
\usepackage{amsmath,amsthm}

\newcommand{\eps}{\varepsilon}

\newcommand{\NP}{{\sf NP}}
\def \P {{\sf P}}
\def \N {{\mathbb N}}

\newcommand{\Ex}{\mathop{\bf E\/}}
\newcommand{\mc}{\mathcal}

\newtheorem{lemma}{Lemma}
\newtheorem{claim}{Claim}
\newtheorem{corollary}{Corollary}
\newtheorem{proposition}{Proposition}

\newtheorem{theorem}{Theorem}
\newtheorem{observation}{Observation}

\newenvironment{reminder}[1]{\smallskip
\noindent {\bf Reminder of #1  }\em}{}

\newenvironment{proofof}[1]{\smallskip
\noindent {\bf Proof of #1.  }}{\hfill$\Box$}

\begin{document}

\title{Improved Parameterized Algorithms for Constraint Satisfaction}

\author{Eun Jung Kim\footnote{This work was performed while the author was at LIRMM-CNRS, supported by the ANR-project AGAPE. Email: {\tt eunjungkim78@gmail.com}}\\ LAMSADE-CNRS, Universit\'{e} de Paris-Dauphine, 75775 Paris, France \and Ryan Williams\footnote{This work was performed while the author was at IBM Research--Almaden, supported by the Josef Raviv Memorial Fellowship. Email: {\tt rrwilliams@gmail.com}} \\ Stanford University, Stanford, CA, USA}

\date{}
\maketitle

\begin{abstract} For many constraint satisfaction problems, the algorithm which chooses a random assignment achieves the best possible approximation ratio. For instance, a simple random assignment for {\sc Max-E3-Sat} allows 7/8-approximation and for every $\eps >0$ there is no polynomial-time ($7/8+\eps$)-approximation unless P=NP. Another example is the {\sc Permutation CSP} of bounded arity. Given the expected fraction $\rho$ of the constraints satisfied by a random assignment (i.e. permutation), there is no $(\rho+\eps)$-approximation algorithm for every $\eps >0$, assuming the Unique Games Conjecture (UGC).

In this work, we consider the following parameterization of constraint satisfaction problems. Given a set of $m$ constraints of constant arity, can we satisfy at least $\rho m +k$ constraint, where $\rho$ is the expected fraction of constraints satisfied by a random assignment? {\sc Constraint Satisfaction Problems above Average} have been posed in different forms in the literature~\cite{Niedermeier2006,MahajanRamanSikdar09}. We present a faster parameterized algorithm for deciding whether $m/2+k/2$ equations can be simultaneously satisfied over ${\mathbb F}_2$. As a consequence, we obtain $O(k)$-variable bikernels for {\sc boolean CSPs} of arity $c$ for every fixed $c$, and for {\sc permutation CSPs} of arity $3$. This implies linear bikernels for many problems under the ``above average'' parameterization, such as {\sc Max-$c$-Sat}, {\sc Set-Splitting}, {\sc Betweenness} and {\sc Max Acyclic Subgraph}. As a result, all the parameterized problems we consider in this paper admit $2^{O(k)}$-time algorithms.

We also obtain non-trivial hybrid algorithms for every Max $c$-CSP: for every instance $I$, we can either approximate $I$ beyond the random assignment threshold in polynomial time, or we can find an optimal solution to $I$ in subexponential time.
\end{abstract}

\thispagestyle{empty}
\newpage
\setcounter{page}{1}

\section{Introduction}

The constraint satisfaction problem~(CSP) is a general language to express many combinatorial problems such as graph coloring, satisfiability and various permutation problems. An instance of a CSP is a set $V$ of variables, a domain $D$ for the variables and ${\mc C}$ a set of constraints. The objective is to assign a value from $D$ to each variable of $V$ so as to maximize the number of satisfied constraints. For example, {\sc 3-Coloring} can be seen as a CSP over a three-element domain, and the constraints correspond to edges (thus the arity of each constraint is 2), indicating that values assigned to the endpoints of an edge must differ. In this work, we are interested in two types of CSPs. In a boolean CSP, the domain $D$ is $\{-1,+1\}$. In a permutation CSP, the size of the domain equals $|V|$ and we request that the assignment is a bijection.

As solving CSPs is $\NP$-hard in general, the next question is whether they allow efficient approximation algorithms. Interestingly, many constraint satisfaction problems exhibit a {\em hardness threshold}, where it is relatively easy to obtain a feasible solution that satisfies a certain fraction of the optimum number of constraints, yet it is difficult to find a solution that is even slightly better. Perhaps the best known example is {\sc Max-E3-Sat}, where we are given a CNF formula in which all clauses have exactly three literals, and wish to find a truth assignment satisfying as many clauses as possible. Although a uniform random assignment satisfies $7/8$ of the clauses, H{\aa}stad~\cite{hastad} proved that it is $\NP$-hard to satisfy $7/8+\eps$ for every $\eps > 0$. The list of problems exhibiting a hardness threshold contains {\sc Max-E$c$-Sat} for $c\geq 3$, {\sc Max-E$c$-Lin-2} for $c\geq 3$, and {\sc E$c$-Set Splitting} for $c\geq 4$. For all these problems, a uniform random assignment achieves the best approximation ratio based on $\P\neq \NP$. Furthermore, for these problems the lower bounds on the optimum are also {\em tight} in the sense that they are optimal for an infinite sequence of instances.

For permutation CSPs, similar results have been identified, conditioned on the Unique Games Conjecture (UGC) of Khot~\cite{khot02}. In the {\sc Betweenness} problem, we have a set of {\em betweenness constraints} of the form ``$v_i$ is between $v_j$ and $v_k$'' for distinct variables $v_i, v_j , v_k \in V$ and the task is to find a permutation of the variables that satisfies the maximum number of constraints. More formally, constraints have the form \[(\pi(v_j) < \pi(v_i) < \pi(v_k)) \vee (\pi(v_k) < \pi(v_i) < \pi(v_j))\] and the task is to find a bijection $\pi : V \rightarrow \{1,\ldots,|V|\}$ that satisfies the maximum. We can satisfy one-third of the constraints in expectation by choosing a uniform random permutation. Moreover, it is hard to achieve a better approximation ratio, assuming UGC~\cite{charikar09}. Hardness thresholds under UGC were known for permutation CSPs of arity 2 and 3~\cite{Guruswami08, charikar09}. Recently these results were generalized to arbitrary fixed arity~\cite{Guruswami11}. Here, the lower bounds on the optimum obtained by a random assignment is tight as well: consider, for example, an instance of {\sc Betweenness} in which we have all possible three constraints for triplets of variables.

These threshold phenomena are fascinating in that they provide a sharp boundary between feasibility and infeasibility: while the average is easy to obtain, satisfying an ``above average" fraction is intractable. While the identification of such thresholds is extremely interesting, it does not end the story but rather initiates  another one. As it is likely that practitioners will require feasible solutions that exceed these easy thresholds, it is important to understand how much computational effort is required to solve a problem beyond its threshold. One way of cleanly formalizing this question uses parameterized complexity.

\begin{comment}
For {\sc Max-E3-Sat}, it is easy to see that a random assignment satisfies $7/8$ of the clauses. An assignment that meets this lower bound can be found in polynomial time \cite{Johnson73}. Since no assignment can satisfy more than $m$ constraints, this implies a $7/8$-approximation algorithm. Interestingly, we cannot hope to improve this approximation ratio unless $\P=\NP$ \cite{hastad}.
\end{comment}

A \emph{parameterized problem} is a subset $L\subseteq \Sigma^* \times
\mathbb{N}$ over a finite alphabet $\Sigma$. $L$ is
\emph{fixed-parameter tractable} if the membership of $(x,k)$ in $\Sigma^* \times \mathbb{N}$ can be decided in time $|x|^{O(1)} \cdot f(k)$ where $f$ is a computable function of the
parameter~\cite{FlumGrohe06}. Given a pair of parameterized problems $L$ and $L'$,
a {\em bikernelization} is a polynomial-time pre-processing algorithm that maps an instance $(x,k)$ to an instance $(x',k')$ (the \emph{bikernel}) such that (i)~$(x,k)\in L$ if and only if
$(x',k')\in L'$, (ii)~ $k'\leq f(k)$, and (iii)~$|x'|\leq g(k)$ for some
functions $f$ and $g$. The function $g(k)$ is called the {\em size} of the bikernel.
A parameterized problem is fixed-parameter tractable if and only if it is decidable and admits a bikernelization~\cite{FlumGrohe06}. A {\em kernelization} of a parameterized problem is a bikernelization to itself. For an overview of kernelization, see the recent survey~\cite{bodlaender}.

\begin{comment}
Note that a {\em standard} parameterization, which asks whether we can find an assignment satisfying at least $k$ constraints, is not interesting for constraint satisfaction problems. Recall that a uniform random assignment satisfies $\rho \cdot m$ constraints for some constant $\rho >0$ depending on the constraint type, and this immediately implies a kernel with linear number of constraints (and thus with a linear number of variables). Indeed, if $\rho\cdot m \geq k$, then the instance is {\em yes} and otherwise, we have at most $k/\rho$ constraints.
\end{comment}

Motivated by the discussion above, our work focuses on the following question:

\begin{quote}
$\text{\sc (permutation) Max-$c$-CSP Above Average}$: We are given a parameter $k$ and a set of constraints with at most $c$ variables per constraint. Each constraint has a positive integer weight. Determine if there is a variable assignment (or permutation) that satisfies a subset of constraints with total weight at least $\rho \cdot W + k$, where $W$ is the total weight of all constraints and $\rho$ is the expected fraction of weighted constraints satisfied by a uniform random assignment.
\end{quote}

\noindent {\bf Previous Work.} Parameterizations \emph{above a guaranteed value} were first considered
by Mahajan and Raman \cite{mahajanJA31} for the problems \textsc{Max-Sat} and \textsc{Max-Cut}. In a recent paper~\cite{MahajanRamanSikdar09}, Mahajan, Raman
and Sikdar argue, in detail, that a practical (and challenging) parameter for a
maximization problem is the number of clauses satisfied above a
\emph{tight} lower bound, which is $(1-2^{-c})m$ for \textsc{Max-Sat} if
each clause contains exactly $c$ different variables. In the monograph by Neidermeier \cite{Niedermeier2006}, an open problem attributed to Benny Chor [25, p.43] asks whether {\sc Betweenness Above Average} is fixed parameter tractable.

A way for systematic investigation of above-average parameterization was recently presented by Gutin et al.~\cite{GutinKimSzeiderYeo09a}. They presented reductions to quadratic bikernels (i.e., bikernels with $O(k^2)$ variables) for the above-average versions of problems such as {\sc Maximum Acyclic Subgraph} and {\sc Max-$c$-Lin-2}. Alon {\em et al.}~\cite{alon2010} pushed forward the idea of representing a CSP instance algebraically, presenting a quadratic kernel for {\sc Max-$c$-CSP} using a similar method. This method was used to give a fixed-parameter algorithm for {\sc Betweenness Above Average} \cite{GutinKMY10}, in which the idea of a {\em coarse} ordering is used. This result was later generalized in \cite{gutinvan} to obtain a quadratic bikernel for {\sc Permutation Max-3-CSP}.

\noindent {\bf Our contribution.} We show that every {\sc Max-$c$-CSP} and {\sc Permutation Max-3-CSP} admits a problem bikernel with only $O(k)$ variables in the above-average parameterization. More precisely, we prove that essentially any {\em hard} instance of {\sc Max-$c$-CSP} must have less than $c(c+1)k/2$ variables; instances with more variables are yes-instances for which good assignments can be generated in polynomial time. This improves over the main results of Alon {\em et al.}~\cite{alon2010} who gave kernels of $O(k^2)$ variables, Crowston {\em et al.}~\cite{swat2010} for {\sc Max-$c$-CSP} with $O(k\log{k})$ variables and \cite{gutinvan} for {\sc Permutation Max-3-CSP} of $O(k^2)$ variables. This implies linear variable bikernels for the above-average versions of many different problems, such as {\sc Max-$c$-Lin-2}, {\sc Max-$c$-Sat}, {\sc Set Splitting} when the sets are of size $\geq 4$, {\sc Maximum Acyclic Subgraph}, {\sc Betweenness}, {\sc Circular Ordering}, and {\sc 3-Linear Ordering}.

The key to our results is a fixed-parameter algorithm for the following problem:

\begin{quote}
$\text{\sc Max-$c$-Lin-$2$ Above Average}$: {\em We are given a parameter $k$ and a system of linear equations over $\mathbb{F}_2$ with at most $c$ variables per equation. Each equation $e$ has a positive integer weight, and the weight of an assignment in the system is defined to be the total sum of weights of equations satisfied by the assignment. Determine if there is an assignment of weight at least $W/2 + k/2$, where $W$ is the total weight of all equations.}
\end{quote}

Note the random assignment algorithm yields $W/2$ weight, and it is famously $\NP$-hard to attain $W/2+\eps W$ for every $\eps > 0$~\cite{hastad}. Our proofs imply the stronger result that {\em every} constraint satisfaction problem admits a hybrid algorithm, in the following sense:

\begin{theorem}\label{hybrid} For {\em every} Boolean {\sc Max-$c$-CSP}, there is an algorithm  with the property that, for every $\eps > 0$, on any instance $I$, the algorithm outputs either:\begin{itemize}
\item an optimal solution to $I$ within $O^{\star}(2^{c(c+1) \eps m/2})$ time, or
\item a $(\rho+\eps/2^c)$-approximation to $I$ within polynomial time, where $\rho$ is the expected fraction of weighted constraints satisfied by a uniform random assignment to the CSP.
\end{itemize}
\end{theorem}

This resolves an open problem of Vassilevska, Williams, and Woo~\cite{vass2006}, who asked if {\sc Max-3-Sat} had an algorithm of this form.

\section{Preliminaries}

We define the {\sc Boolean Max-$c$-CSP} and {\sc Max-$c$-Permutation CSP}. A boolean constraint satisfaction problem is specified by the domain $\{-1,+1\}$ and a set of predicates $\mathbb{P}$, called {\em payoff functions} as well. A predicate $P\in \mathbb{P}$ is a function from $\{-1,+1\}^{c'}$ to $\{0,1\}$ for $c'\leq c$. The maximum number $c$ of inputs to the predicates in $\mathbb{P}$ is the arity of the problem. We interpret $-1$ as the value {\sf True} and $+1$ as {\sf False}.

An instance of boolean CSP is specified as a set of variables $V$ along with a collection of triples $\mc{C}$ = $\{(f_1,S_1,w_1) ,\ldots, (f_m,S_m,w_m)\}$, where $f_i\in \mathbb{P}$, every $S_i$ is an ordered tuple from $V$ of size at most $c$, every $w_i$ is a positive integer. A variable assignment $\phi : V \rightarrow D$ {\em satisfies} a constraint $(f_i,(s_1,\ldots,s_{c'}),w_i)$ provided that $f(\phi(s_1),\ldots , \phi(s_{c'})) = 1$. Our goal is to find an assignment $\phi: V\rightarrow \{-1,+1\}$ of maximum weight. Here the {\em weight of an assignment $\phi$} is defined to be the total sum of weights $w_i$ of constraints $(f_i,S_i,w_i)$ satisfied by $\phi$. For example, {\sc Max-E3-Sat} is specified by a single predicate $P: \{-1,+1\}\rightarrow \{0,1\}$, where $P(x,y,z)=0$ if and only if $x=y=z=+1$.

We consider the following parameterization of boolean CSP.

\begin{quote}
$\text{\sc Max-$c$-CSP Above Average ($c$-CSP$_{AA}$)}$\\
{\em Input:} A set $V$ of variables, a collection of constraints $\mc{C}=\{(f_1,S_1,w_1),\ldots, (f_m,S_m,w_m)\}$ with $|S_i|\leq c$ for every $i$, an integer $k\geq 0$.\\
{\em Parameter:} $k$\\
{\em Goal:} Determine if there is an assignment with weight at least $\rho \cdot W + k/2^c$, where $W$ is the total weight of all constraints and $\rho$ is the expected fraction of weighted constraints satisfied by a uniform random assignment.
\end{quote}

In the permutation CSP problem, the domain $D$ is $[n]=\{1,\ldots ,n\}$ and a predicate $P\in \mathbb{P}$ is a function from $\mc{S}_{c'}$ to $\{0,1\}$ for $c'\leq c$, where $\mathbb{S}_{c'}$ is the set of permutations on $\{1,\ldots , c'\}$. Let $\mathbb{P}$ be a set of predicates. An instance of permutation CSP is given as a variable set $V$ and a collection $\mc{C}=\{(f_1,S_1,w_1),\ldots, (f_m,S_m,w_m)\}$, where every $S_i$ is an ordered tuple from $V$ of size at most $c$, every $w_i$ is a positive integer, $f_i\in \mathbb{P}$ is applied to the tuple $S_i$. In a permutation CSP, a variable assignment $\phi : V \rightarrow D$ is required to be a bijection, or equivalently, a permutation. A permutation $\phi$ {\em satisfies} a constraint $(f_i,(s_1,\ldots,s_{c'}),w_i)$ provided that $f(\phi(S_i)) = 1$, viewing the local permutation $\phi(S_i)$ as an element of $\mathbb{S}_{c'}$. Our goal is to find an assignment $\phi: V\rightarrow [n]$ of maximum weight, where the weight of $\phi$ is the total sum of weights of constraints satisfied by $\phi$. For example, {\sc Betweenness} is specified by a single predicate of arity 3, i.e. $P:\mathbb{S}_3 \rightarrow \{0,1\}$, where $P(x,y,z)=1$ if and only if $xyz \in \{123,321\}$.

We consider the following parameterization of permutation CSP:

\begin{quote}
$\text{\sc Max-$c$-Permutation CSP Above Average}$\\
{\em Input:} A set $V$ of variables, a collection of constraints $\mc{C}=\{(f_1,S_1,w_1),\ldots, (f_m,S_m,w_m)\}$, an integer $k\geq 0$.\\
{\em Parameter:} $k$\\
{\em Goal:} Determine if there is an assignment with weight at least $\rho \cdot W + k/(c!4^c)$, where $W$ is the total weight of all constraints and $\rho$ is the expected fraction of weighted constraints satisfied by a uniform random permutation.
\end{quote}

In what follows, we omit {\sc Above Average} and simply say {\sc Max-$c$-Lin-2}, {\sc Max-$c$-CSP} and {\sc Max-$c$-Permutation CSP} to refer to the parameterized problems.

For {\sc Max-$c$-Permutation CSP}, let $\mathbb{P}$ be the associated set of predicates. For each predicate $P \in \mathbb{P}$, we can identify the set $\Pi_P=\{\pi \in \mathbb{S}_c: P(\pi)=1\}$. Notice that $\phi$ satisfies $(f_i,S_i=(v_1,\ldots, v_{c'}),w_i)$ if and only if there is a $\pi \in \Pi_{f_i}$ s.t. $\phi(v_1)\phi(v_2)\cdots \phi(v_{c'})\cong \pi(1)\pi(2)\cdots \pi(c')$, and thus if and only if there is $\pi \in \Pi_{f_i}$ s.t. $\phi(v_{\pi^{-1}(1)}) < \phi(v_{\pi^{-1}(2)}) < \cdots < \phi(v_{\pi^{-1}(c')})$. The folklore result below allows us to focus on the case when $\mathbb{P}$ contains a single predicate $P$ such that $\Pi_P=\{12\cdots c\}$, which we denote {\sc $c$-Linear Ordering}. The proof for the case $c = 3$ and $|\mathbb{P}|=1$ can be found in \cite{gutinvan} and its extension for arbitrary fixed $c$ is straightforward. As the proof of \cite{gutinvan} considers only the case when $\mathbb{P}$ has a single predicate, we sketch the proof here even though its generalization is immediate.

\begin{proposition}\label{focuslin}
Let $(V,\mc{C},k)$ be an instance of {\sc $c$-OCSP$_{AA}$}. There is a polynomial time transformation $R$ from {\sc Max-$c$-Permutation CSP} to {\sc $c$-Linear Ordering} such that an instance $(V,\mc{C},k)$ of {\sc Max-$c$-Permutation CSP} is a yes-instance if and only if $R(V,\mc{C},k)$ is a yes-instance of {\sc $c$-Linear Ordering}.
\end{proposition}
\begin{proof}
From an instance $(V,\mc{C},k)$ of {\sc Max-$c$-Permutation CSP}, we construct an instance $(V,\mc{C}_0,k)$ of {\sc $c$-Linear Ordering} as follows. We shall express a constraint $(f_i,S_i,w_i)$ by a set of constraints $R(f_i,S_i,w_i)$ such that $(f_i,S_i,w_i)$ is satisfied if and only if exactly one of $R(f_i,S_i,w_i)$ is satisfied. Let $S_i$ be $(v_1,v_2, \ldots , v_{c'})$.

For every element $\pi$ of $\Pi_{f_i}=\{\pi \in \mathbb{S}_c: f_i(\pi)=1\}$, we add to the instance of {\sc $c$-Linear Ordering} the constraint $(P_{id}, (v_{\pi^{-1}(1)}, v_{\pi^{-1}(2)},\ldots , v_{\pi^{-1}(c')}),w_i)$. Here $P_{id}$ is a function mapping identity permutation to 1 and other permutations to 0. Notice that $(f_i,S_i,w_i)$ is satisfied if and only if exactly one of the constraints $(P_{id}, (v_{\pi^{-1}(1)}, v_{\pi^{-1}(2)},\ldots , v_{\pi^{-1}(c')}),w_i)$ for $\pi \in \Pi_{f_i}$ is satisfied. Hence the weight of a linear ordering $\phi$ remains the same in the original and transformed instances. Moreover, the expected satisfied fraction of the constraint $(f_i,S_i,w_i)$ is $|\Pi_{f_i}|/(c')!$ and the expected satisfied fraction of the new constraints is the same. Hence, the instance of {\sc Max-$c$-Permutation CSP} has a linear ordering of weight $\rho\cdot W + k$ if and only if the constructed instance of {\sc $c$-Linear Ordering} a linear ordering of weight $\rho\cdot W + k$.
\end{proof}

It is well-known that for every function $f:\{1,-1\}^n \rightarrow \mathbb{R}$ can be uniquely expressed as a multilinear polynomial $$ f(x)=\sum_{S\subseteq [n]} \hat{f}(S)\chi_S(x),$$ where $\hat{f}(S)$ is the fourier coefficient of $f$ on $S$, defined as $$\hat{f}(S) := \Ex_{x\in \{1,-1\}^n} [f(x)\chi_S(x)]$$ and the character function $\chi_S(x)$ is defined as $\chi_S(x) :=\prod_{i\in S}x_i$. Given the truth table of $f$, the fourier coefficients of $f$ can be computed via the inverse fourier transform~(one reference is \cite{ODonnell08}).

\section{Max-$c$-Lin-2 Above Average}\label{secaver}

We now turn to describing improved parameterized algorithms for maximum constraint satisfaction problems with a constant number of variables per constraint, including the problems of satisfying a maximum subset of linear equations and maximum CNF satisfiability. At the heart of our approach is a faster algorithm for $\text{\sc Max-$c$-Lin-2 Above Average}$ that can be applied in a general way to solve other CSPs.

\begin{theorem}\label{maxclin2} For every $c \geq 2$,  {\sc Max-$c$-Lin-$2$ Above Average} can be solved in $O(2^{(c(c+1)/2)k} \cdot m)$ time. \end{theorem}

In \cite{vass2006}, the authors gave a ``hybrid algorithm'' for the unweighted problem {\sc Max-E3-Lin-$2$} (where exactly three variables appear in each equation), with the property that, after a polynomial time test of the instance, the algorithm either outputs an assignment satisfying $(1/2 + \eps) m$ equations in polynomial time, or outputs the optimal satisfying assignment in $2^{O(\eps m)}$ time. The algorithm works by finding a maximal subset of equations such that every pair of equations share no variables; based on the size of this set, the hybrid algorithm decides to either approximately solve the instance or solve it exactly. Our algorithm is in a similar spirit, but requires several modifications to yield a parameterized algorithm for the weighted case, to deal with any $c \geq 2$, and to deal with ``mixed'' equations that can have different numbers of variables.

Let $F$ be a set of equations over $\mathbb{F}_2$, where each equation $e$ contains at most $c$ variables and has a positive integral weight $w(e)$. For a single equation $e \in F$, let $var(e)$ be the set of all variables appearing in $e$. Let $var(F) = \bigcup_{e \in F} var(e)$. For a set of equations $F'$, the weight $w(F')$ is the sum of weights $w(e)$ over $e\in F'$. The {\em weight of an assignment} is the total weight of equations that are satisfied by the assignment.

Note that {\sc Max-2-Lin-2 Above Average} is a generalization of {\sc Max Cut Above Average} on weighted graphs: by simulating each edge $\{u,v\}$ of weight $w$ with an equation $x_u + x_v = 1$ of weight $w$, the {\sc Max-2-Lin-2} problem easily captures {\sc Max Cut}.

We assume that the given instance is {\em reduced} in the sense that there is no pair of equations $e,e'$ with $e\equiv e'+1 \pmod 2$. (Such an equation $e$ is said to be {\em degenerate} in \cite{vass2006}.) If such a pair exists, one can remove the equation of lesser weight (call it $e'$) and subtract $w(e')$ from $w(e)$. Note the weight of every variable assignment has now been subtracted by $w(e')$.

\begin{proofof}{Theorem~\ref{maxclin2}} It is convenient to view an equation $e$ as a set $var(e)$. We first find a maximal independent (i.e. disjoint) collection $S_c \subseteq F$ of $c$-sets. More precisely, we treat each equation as a set of variables, ignore those sets of cardinality less than $c$, and find a maximal disjoint set over the $c$-sets using the standard greedy algorithm. All remaining equations in $F$ now have at most $c-1$ variables if we remove all occurrences of variables in $var(S_c)$ from $F$.

Next, we pick another collection $S_{c-1}\subseteq F$ of sets with the property that, after we remove all variables in $var(S_c)$ from $F$, $S_{c-1}$ forms a maximally independent collection of $(c-1)$-sets in the remaining set system. In general, for $j=c-2$ down to $1$, once the variables in $var(S_c \cup \cdots \cup S_{j+1})$ have been removed from the remaining equations, a maximal independent set of $j$-sets is chosen greedily, and we set $S_j$ to be a collection of corresponding original sets in $F$ (with the variables in $var(S_c \cup \cdots \cup S_{j+1})$ added back). We continue until $S_1$, in which each set in the collection has exactly one variable after those in $var(S_c \cup \cdots \cup S_2)$ have been removed. For convenience, let $S_{c+1}=var(S_{c+1})=\emptyset$. By properties of maximal disjoint sets, we have:

\begin{observation}\label{atmostj}
For every $1\leq j \leq c$, eliminating the variables appearing in $var(S_c \cup \cdots \cup S_{j+1})$ leaves at most $j$ variables in every equation of $F$.
\end{observation}

Now, either {\bf (1)} $w(S_j) < k$ for every $j=1,\ldots,c$, or {\bf (2)} there is a $j$ such that $w(S_j) \geq k$.

Case {\bf (1)} is easily handled: for every $j$, each equation in $S_j$ contains $j$ variables which do not appear in $S_c \cup \cdots \cup S_{j+1}$. Hence, $|var(F)|=|var(\bigcup_{i=1}^c S_i)|< ck+(c-1)k + \cdots + k < (c(c+1)/2)\cdot k$. By trying all $O(2^{(c(c+1)/2)k})$ assignments to $var(F)$, we can find an optimal assignment for $F$.

Case {\bf (2)} is more delicate and is handled by the two claims below. We will show that in this case, $F$ is a yes-instance of the problem and we can efficiently recover a solution for it. Recall an equation $e \in F$ is non-degenerate if there is no $e' \in F$ such that $e \equiv e' + 1 \pmod 2$. As mentioned earlier, we may assume without loss of generality that every equation in $F$ is non-degenerate.

\begin{claim}\label{maxclin2clm2}
For every $1\leq j \leq c$, a random assignment satisfying all equations in $S_j$ will satisfy every non-degenerate equation in $F-S_j$ with probability 1/2. Moreover, we can output such a random assignment in polynomial time.
\end{claim}

\proof To prove the first part of the claim, it suffices to show that no equation $e\in F-S_j$ (or its negation $e+1$) can be expressed as a linear combination of one or more equations in $S_j$. Put another way, we will show that every equation in $e \in F-S_j$ is {\em linearly independent} of the equations in $S_j$.

Suppose there are equations $e_1,\ldots,e_m$ from $S_j$ such that their summation (modulo 2) results in a variable subset that is equal to the set of variables in another equation $e\in F$. That is, viewing $e_1,\ldots,e_m$ and $e$ as indicator $n$-bit vectors (one bit for each of the $n$ variables, omitting the constant terms in the equations), we have $e = \sum_{i=1}^m e_i \pmod 2$. Recall that every equation in $S_j$ has $j$ variables which do not appear in $var(S_c \cup \cdots \cup S_{j+1})$, and every pair of equations in $S_j$ involves disjoint sets of variables, by construction. Hence, if $m>1$, then the equation $e$ (composed of variables from $e_1,\ldots,e_m$) has more than $j$ variables which do not appear in $var(S_c \cup \cdots \cup S_{j+1})$, which is impossible by Observation~\ref{atmostj}. Therefore $m=1$, and every subset $\{e_1,\ldots, e_m\}$ of equations from $S_j$ whose modulo sum is the same as another equation $e\in F-S_j$ has cardinality $1$. But then the equation $e$ is degenerate, which is a contradiction to the non-degeneracy assumption. Therefore no non-degenerate equation in $F$ (or its negation) can be represented as a linear combination of one or more equations from $S_j$.

Now, given that every non-degenerate equation in $e \in F-S_j$ is linearly independent of the equations in $S_j$, we claim that a random assignment that is consistent with the equations in $S_j$ will satisfy $e$ with probability $1/2$. This is a simple consequence of linear algebra over $\mathbb{F}_2$. Put the system of equations $S_j$ in the form $Ax = b$, where $A \in \mathbb{F}_2^{|S_j| \times n}$, $x \in \mathbb{F}_2^n$, and $b \in \mathbb{F}_2^{|S_j|}$. Let $e \in F-S_j$. Define $B_e \in \mathbb{F}_2^{(|S_j|+1) \times n}$ to be identical to $A$ in its first $|S_j|$ rows, and in the last row, $B_e$ contains the indicator vector for the variables of $e$. Define $c_e \in \mathbb{F}_2^{|S_j|+1}$ to be identical to $b$ in its first $|S_j|$ components, and $c$ contains the constant term of $e$ in its last component. Saying that $e \in  F-S_j$ is linearly independent of $S_j$ is equivalent to saying $rowrank(B_e) = rowrank(A)+1$, and the set of solutions to $Ax = b$ contains the set of solutions to $B_e x = c_e$. The number of solutions to a system of rank $r$ is $2^{n-r}$. Therefore a uniform random variable assignment that satisfies $Ax = b$ will also satisfy $B_e x = c_e$ with probability $1/2$.

Finally, we describe how to produce a uniform random assignment over all assignments that satisfy the equations in $S_j$. Produce a random assignment to the variables in $var(S_{c+1}\cup \cdots \cup S_{j+1})$, then produce a random assignment to those variables in the maximal independent collection of $j$-sets obtained after removing $var(S_{c+1}\cup \cdots \cup S_{j+1})$, in such a way that every equation in $S_j$ is satisfied. (Exactly one variable in each equation of $S_j$ will be ``forced'' to be a certain value, but note that none of these forced variables appear in more than one equation of $S_j$, by construction.) The remaining variables are set to 0 or 1 uniformly at random. Note that if $j=1$ and some equation $e\in S_1$ has $|var(e)|=1$, the assignment to the variable of $e$ is decided uniquely.\qed

\begin{claim}\label{maxclin2clm1}
If there is a $j$ with $w(S_j)\geq k$, then we can find an assignment with weight at least $W/2 + k/2$ in polynomial time.
\end{claim}

\proof Suppose that $j\geq 1$ is the largest integer with $w(S_j)\geq k$. By Claim~\ref{maxclin2clm2}, a random assignment satisfying all equations in $S_j$ will satisfy every other non-degenerate equation with probability 1/2. Hence the weight of such an assignment is at least $(W-w(S_j))/2+w(S_j)\geq W/2+k/2$ on average. An assignment can also be found deterministically using conditional expectation.
\qed

This completes the proof of Theorem~\ref{maxclin2}.\end{proofof}

The above proof shows that the following stronger statement is also true.

\begin{theorem}\label{maxclin2yes} For every $c \geq 2$, let $I$ be an instance of {\sc Max-$c$-Lin-$2$ Above Average}. If $|var(F)|\geq (c(c+1)/2)k$, then $I$ is an yes-instance and an assignment satisfying equations with at least $W/2 + k/2$ weight can be found in polynomial time.\end{theorem}

%%%%%%%%%%%%%%%%%%%%%%%%%%%%%%%%%%%%%%%%%%%%%%%%%%%%%%%%%%%%%%%%%%%%%%%%%%%%%%%%

Observe that the running time of our algorithm is optimal up to constant factors in the exponent, assuming the Exponential Time Hypothesis:

\begin{theorem}\label{eth} If {\sc Max-$3$-Lin-2 Above Average} can be solved in $O(2^{\eps k} 2^{\eps m})$ time for every $\eps > 0$, then 3SAT can be solved in $O(2^{\delta n})$ time for every $\delta > 0$, where $n$ is the number of variables. \end{theorem}

\begin{proof} First, by the improved Sparsification Lemma of~\cite{calabro2006}, for every $\delta > 0$ we can reduce 3SAT on $n$ variables and $m$ clauses in $2^{\delta n}$ time to 3SAT on $n$ variables and $m'=(1/\delta)^{c} n$ clauses, for some fixed constant $c > 1$. This 3SAT instance on $n$ variables and $m'$ clauses can further be reduced to Max-$3$-Lin-2 on $n$ variables and $O(m')$ clauses using the reduction of Lemma~\ref{csp2lin} (proved below). Provided that we can determine whether $m'/2+k/2$ equations can be satisfied in $2^{\eps k} 2^{\eps m'}$ time, then by trying each $k$ in the interval $[1,m']$ we can solve the Max-$3$-Lin-2 instance exactly in at most $2^{2\eps m'} \leq O(2^{2\eps (1/\delta)^c n})$ time.

This results in an $O(2^{\delta n + 2\eps (1/\delta)^c n})$ algorithm for 3SAT. Setting $\eps = \delta^{c+1}$, we obtain $O(2^{3 \delta n})$ time. As this reduction works for every $\delta > 0$, the conclusion follows. \end{proof}

%\section{Applications to Boolean and Permutation CSPs}

\section{Boolean MAX-$c$-CSP Above Average}

%%%%%%%%%%%%%%%%%%%%%%%%%%%%%%%%%%%%%%%%%%%%%%%%%%%%%%%%%%%%%%%%%%%%%%%%%%%%%%%%

To apply our algorithm to general CSPs, we use the following reduction.

\begin{lemma}[\cite{alon2010},\cite{swat2010}] There is a polynomial time reduction from {\sc Max-$c$-Csp Above Average} with $n$ variables and parameter $k$ to {\sc Max-$c$-Lin-2 Above Average} with $n$ variables and parameter $k$. \label{csp2lin}
\end{lemma}

The proof of the lemma for unweighted $c$-CSP is sketched in~\cite{alon2010}, and a full proof is given in~\cite{swat2010}. Here we give an alternative proof which also covers the weighted case. (Although we are confident that the proofs in~\cite{alon2010,swat2010} also extend to the weighted case, we include a proof in Appendix~\ref{appendixproof} for completeness.)

%Let us sketch the proof of Lemma~\ref{csp2lin} for completeness.

\begin{theorem}\label{maxccsp} For every $c \geq 2$, {\sc Max-$c$-CSP Above Average} can be solved in $O(2^{(c(c+1)/2)k} \cdot m)$ time.\end{theorem}

\begin{proofof}{Theorem~\ref{maxccsp}} Using the reduction of Lemma~\ref{csp2lin}, reduce an instance of {\sc Max-$c$-CSP Above Average} with $m$ constraints to {\sc Max-$c$-Lin-2 Above Average} with $O(2^c \cdot m)$ equations. Using the algorithm of Theorem~\ref{maxclin2} we solve the obtained instance of {\sc Max-$c$-Lin-2 Above Average}. Thus we can determine if the given $c$-CSP has an assignment with weight at least $AVG + k$ in $O(2^{(c(c+1)/2)k}\cdot 2^c \cdot m) = O(2^{(c(c+1)/2)k}\cdot m)$ time. To finding an actual solution for {\sc Max-$c$-CSP Above Average}, we can simply use the transformation given in the proof of Lemma~\ref{csp2lin}.
\end{proofof}

Theorem~\ref{maxclin2yes} and Lemma~\ref{csp2lin} show in fact that every CSP admits a hybrid algorithm~\cite{vass2006}.

\begin{reminder}{Theorem~\ref{hybrid}} For every Boolean {\sc Max-$c$-CSP}, there is an algorithm  with the property that, for every $\eps > 0$, on any instance $I$, the algorithm outputs either:\begin{itemize}
\item an optimal solution to $I$ within $O^{\star}(2^{c(c+1) \eps m/2})$ time, or
\item a $(\rho+\eps/2^{c+1})$-approximation to $I$ within polynomial time, where $\rho$ is the expected fraction of weighted constraints satisfied by a uniform random assignment to the CSP.
\end{itemize}
\end{reminder}

\begin{proof} Given an instance $I$ of {\sc Max-$c$-CSP} with $m$ constraints, Lemma~\ref{csp2lin} shows that we can reduce $I$ to an weighted instance $I'$ of {\sc Max-$c$-Lin-2} with $O(2^c m)$ equations, in polynomial time, such that at least $\rho m + \delta k$ constraints can be satisfied in $I$ if and only if at least $W/2 + k/2$ weight of equations can be satisfied in $I'$, where $\delta \geq 1/2^c$ and depends on the underlying constraints. Now set $k = \eps m$ and run the algorithm of Theorem~\ref{maxclin2yes}.  If $|var(I')|\geq (c(c+1)/2)k$, then an assignment satisfying at least $W/2 + k/2 \geq W/2 + \eps m/2$ weight of equations can be found in polynomial time, hence we obtain an assignment for $I$ satisfying at least $\rho m + \delta \eps m/2$ constraints. Otherwise, exhaustive search over the $c(c+1)\eps m/2$ variables of $I'$ will uncover an exact solution to $I$ in $O^{\star}(2^{c(c+1) \eps m/2})$ time. \end{proof}

%{\bf \texttt{better delete this?} Q: Can we get a result for {\sc Max-Lin-2}? check if $m/2+ k$ equations can be satisfied, in FPT?}

We close this section with showing how our algorithm can provide linear size kernels for {\sc Max-$c$-Lin-2 Above Average} and a kernel for {\sc Max-$c$-CSP Above Average}.

\begin{corollary}\label{maxclin2ker}
For every $c\geq 3$, the problem {\sc Max-$c$-Lin-2 Above Average} can be reduced to a problem kernel with at most $(c(c+1)/2)k$ variables in polynomial time.
\end{corollary}

\begin{proof} Consider executing the algorithm of Theorem~\ref{maxclin2}, up to the point before it performs an exhaustive search of assignments. At this point, the algorithm has taken only polynomial time. If there is a $S_j$ with weight at least $k$, the algorithm outputs an assignment with weight at least $W/2+k/2$ in polynomial time. Otherwise, for all $j=c,c-1,\ldots,1$,  $S_j$ has weight less than $k$. It follows (from Case {\bf (1)} in the proof of Theorem~\ref{maxclin2}) that the total number of variables in the instance is at most $(c(c+1)/2)k$.\end{proof}

Note that the size of the kernel in Corollary~\ref{maxclin2ker} matches the prior work for $c=2$~\cite{alon2010}.

\begin{corollary}\label{maxccspker}
For every $c\geq 3$, the problem {\sc Max-$c$-CSP Above Average} can be reduced to a problem kernel with at most $(c(c+1)/2)k$ variables in polynomial time.
\end{corollary}
\begin{proof}
In the proof of Theorem 1 in~\cite{alon2010}, a procedure ${\cal P}$ is given that reduces any instance of {\sc Max-$c$-Lin-2 Above Average} with total sum of weights $W$ and parameter $k$ into an instance of {\sc Max-$c$-CSP Above Average}\footnote{In \cite{alon2010}, the transformed instance is in fact a {\sc Max-$c$-SAT} instance.} with (a multiset of) $2^{c-1}W$ constraints and parameter $2^{c-1}k$. More precisely, the procedure ${\cal P}$ considers an instance of {\sc Max-$c$-Lin-2 Above Average} in which each equation has weight 1 and the multiplicity of an equation may be larger than one. ${\cal P}$ maps an equation into a set of $2^{c-1}$ clauses.

Given a {\sc Max-$c$-CSP Above Average} instance on $n$ variables and $m$ constraints, we first perform the transformation given by Lemma~\ref{csp2lin} and obtain a {\sc Max-$c$-Lin-2 Above Average} instance, with $O(2^c m)$ equations and $n$ variables. By applying the kernelization of Theorem~\ref{maxclin2} we obtain an equivalent instance with at most $(c(c+1)/2)k$ variables and no more than $(c(c+1)k/2)^c$ (weighted) equations. Finally, apply procedure ${\cal P}$ to reduce the problem {\em back} into a {\sc Max-$c$-CSP Above Average} instance, having $(c(c+1)/2)k$ variables, $O(2^{c-1}\cdot (c(c+1)k/2)^c)$ constraints and parameter $2^{c-1}k$. \end{proof}

\section{MAX-$c$-Permutation CSP Above Average}\label{sec:per}

In Section \ref{secaver}, we proved that every reduced instance of {\sc Max-$c$-Lin-2} either has an assignment satisfying $W/2 + k$ weight of equations, or has at most $(c(c+1)/2)k$ variables. In this section, we describe how to apply this result to the problems {\sc Max-$c$-Permutation CSP} for $c=2,3$ and obtain a bikernel with $O(k)$ variables. By Proposition \ref{focuslin}, it suffices for us to focus on the problem {\sc 3-Linear Ordering} instead of considering general {\sc Max-$c$-Permutation CSP}. The outline of our proof is as follows. We first introduce new reduction rule for simplifying 3-linear ordering instances. Then we reduce a given instance to another intermediate problem called {\sc $t$-Ordering}. Finally we show how to reduce {\sc $t$-Ordering} to {\sc Max-$3t$-Lin-2}, so we can apply our results from the previous section.

\begin{comment}
To this end, we establish a new link from {\sc Max-$c$-Lin-2} to {\sc $c$-OCSP} in this section. We introduce some reduction rules for {\sc $c$-OCSP} and use the idea of {\em bucket ordering} to transform the problem into boolean CSP. To approximate a linear ordering using a bucket ordering is a natural approach for permutation constraint constraint problems, for example \cite{GutinKMY10, gutinvan, charikar09}. The key ingredient of our work is to reduce the {\sc $3$-Linear Ordering} instance via the reduction rules so that {\em every variable in the reduced {\sc $3$-Linear Ordering} instance has a representative in the associated {\sc $t$-Ordering} instance.} Based on this observation, we can safely handle {\sc $t$-Ordering} instance as a boolean CSP instance, in particular {\sc Max-$c$-Lin-2 Above Average} instance. In particular, the result from the previous section implies that the reduced {\sc $3$-Linear Ordering} instance is a kernel with $O(k)$ variables.
\end{comment}

Recall that in {\sc 3-Linear Ordering}, we are given a collection $\mc{C}$ of constraints $e_i$ over a variable set $V$ with weights $w_i$, sometimes denoted as $w_i(e_i)$, and an integer $k\geq 0$. A constraint $e_i$ is an ordered tuple of size at most three, i.e. $e_i=(v_1, v_2 , v_3)$ or $e_i=(v_1,v_2)$. Constraint $e_i$ is satisfied by a linear ordering $\phi$ if $\phi(v_1) < \phi(v_2) < \phi(v_3)$ (or $\phi(v_1) < \phi(v_2)$, respectively). Let $w(\phi, \mc{C})$ be the total sum of weight of constraints satisfied by $\phi$. Our task is to decide if there is a linear ordering $\phi:V\rightarrow [n]$ of weight $w(\phi,\mc{C})$ at least $\rho\cdot W+k$, where $W$ is the total weight of constraints in $\mc{C}$ and $\rho$ is the expected fraction of weighted constraints satisfied by a uniform random assignment.

For an ordering constraint $e=(v_1,v_2,v_3)\in \mc{C}$, let $var(e) = \{v_1,v_2,v_3\}$ and let $e(i)$ denote the $i$-th cordinate of $e$. Let $\Delta(u,v,w)$ be the set of all possible six constraints on variables $u,v,w$. Let $v_1,v_2,v_3$ be a triple of variables and suppose $\Delta(v_1,v_2,v_3)\cap \mc{C}\neq \emptyset$. Then $\mc{C}(u,v,w)$ is defined as $\Delta(v_1,v_2,v_3)\cap \mc{C}$. Given a constraint $e=(v_1,v_2,v_3)$, we simplify $\Delta(v_1,v_2,v_3)$ and $\mc{C}(v_1,v_2,v_3)$ as $\Delta(e)$ and $\mc{C}(e)$ respectively.

We say that two distinct constraints $e,e'\in \Delta(v_1,v_2,v_3)$ are a {\em reverse pair} if the ordering of $e'=(e(3),e(2),e(1))$ or equivalently, $e(2)=e'(2)$. We say that the reverse pair $e,e' \in \Delta(v_1,v_2,v_3)$ is {\em absent} in the instance $I=(V,\mc{C},k)$ if none of $e$ and $e'$ appear in $\mc{C}$.

We have the following reduction rules for {\sc 3-Linear Ordering} which were known in the literature and their correctness is obvious.

\noindent {\bf Redundancy Rule \:} Remove a variable $v$ from $V$ if it does not appear in any constraint. Remove a constraint $e$ from $\mc{C}$ if its weight is zero.

\noindent {\bf Merging Rule:} If $e_1$ and $e_2$ are identical, then replace them by a single constraint $e_1$ of weight $w_1+w_2$.

\noindent {\bf Cancellation Rule:} If there are two constraints $e_1$, $e_2$ with $|e_1|=|e_2|=2$ and $e_2=(e_1(2),e_1(1))$, let $w_{min}=\min\{w(e_1), w(e_2)\}$ and replace the weights by $w'(e_1):=w(e_1)-w_{min}$ and $w'(e_2):=w(e_2)-w_{min}$.

Here we introduce two new reduction rules.

\noindent {\bf Edge Replacement Rule:} If $e_1,e_2,e_3$ are three constraints in $\mc{C}$ with $var(e_1)=var(e_2)=var(e_3)$ and such that $e_2=(e_1(2),e_1(1),e_1(3))$ and $e_3=(e_1(1),e_1(3),e_1(2))$, then:
\begin{itemize}
\item replace the weight of a constraint by $w'(e_i):=w(e_i)-w_{min}$ for each $i=1,2,3$, where $w_{min}=\min\{w(e_1), w(e_2), w(e_3)\}$.
\item add the binary ordering constraint $(e_1(1),e_1(3))$ of weight $w_{min}$.
\end{itemize}

\noindent {\bf Cycle Replacement Rule:} If $e_1,e_2,e_3$ are three constraints in $\mc{C}$ with $var(e_1)=var(e_2)=var(e_3)$ and such that $e_2=(e_1(2),e_1(3),e_1(1))$ and $e_3=(e_1(3),e_1(1),e_1(2))$, then:
\begin{itemize}
\item replace the weight of a constraint by $w'(e_i):=w(e_i)-w_{min}$ for each $i=1,2,3$, where $w_{min}=\min\{w(e_1), w(e_2), w(e_3)\}$.
\item add the three binary ordering constraints $(e_1(1),e_1(2))$, $(e_1(2),e_1(3))$ and $(e_1(3),e_1(1))$, each of weight $w_{min}$.
\end{itemize}

\begin{lemma}\label{sound}
The Edge Replacement and Cycle Replacement Rules are sound.
\end{lemma}
\begin{proof}
Let $(V,\mc{C},k)$ be an instance of {\sc 3-Linear Ordering} and $(V,\mc{C}',k)$ be an instance obtained by applying one of the reduction rules. To see that Edge Replacement Rule is sound, we observe that $w(\phi, \mc{C})=w(\phi, \mc{C}')$ and the claim follows. For the Cycle Replacement Rule, $w(\phi, \mc{C})=w(\phi, \mc{C}')-w_{min}$ as any ordering $\phi$ satisfies at least one of the three binary ordering constraints added and $\phi$ satisfies two of them in $\mc{C}'$ if and only if it satisfies one of three cyclic order constraints in $\mc{C}$. Therefore, there exists a $\phi$ such that $w(\phi,\mc{C})\geq \rho\cdot W +k$ if and only if $w(\phi, \mc{C}')\geq (\rho\cdot W + w_{min}) + k$. It remains to observe that the average weight of $w(\phi,\mc{C})$ increases by $w_{min}$ by the reduction, which completes the proof.
\end{proof}

Note that we can apply all the above reduction rules exhaustively in polynomial time.  We say that an instance $I=(V,\mc{C},k)$ is {\em irreducible} if none of the above five reduction rules can be applied. We are ready to state the main result of this section.

\begin{theorem}\label{linker3per}
Let $I=(V,\mc{C},k)$ be an irreducible instance of {\sc 3-Linear Ordering}. If $I$ is a no-instance (i.e., less than $\rho W + k$ constraints in $I$ can be simultaneously satisfied), then the number of variables in $I$ is less than $15k$ variables.
\end{theorem}

The rest of this section will be devoted to proving Theorem \ref{linker3per}. We start with an observation about the irreducible instances.

\begin{lemma}\label{absentpair}
Let $I=(V,\mc{C},k)$ be an irreducible instance of {\sc 3-Linear Ordering}. Then for every triplet $v_1,v_2,v_3$ with $\mc{C}(v_1,v_2,v_3)\neq \emptyset$, there is an absent pair $e,e'\in \Delta(v_1,v_2,v_3)$ in $I$.
\end{lemma}

\begin{proof}
Suppose that $I$ has no absent pair in $\Delta(v_1,v_2,v_3)$. Then for each $v_i$, there exists a constraint $e_i\in \mc{C}(v_1,v_2,v_3)$ with $e_i(2)=v_i$. Consider the set $\{e_1(1),e_2(1),e_3(1)\}$ consisting of the first variables of $e_1,e_2$ and $e_3$. There are two possibilities: (a) $\{e_1(1),e_2(1),e_3(1)\}=\{v_1,v_2,v_3\}$ and (b) $\{e_1(1),e_2(1),e_3(1)\}=\{v_1,v_2,v_3\}\setminus \{v_i\}$ for some $i$. One can verify that the Cycle Replacement Rule applies in the case (a) and the Edge Replacement Rule applies in the case of (b).
\end{proof}

We now reduce 3-linear ordering to another problem that maps variables to bit vectors. This representation will be useful in the final reduction to {\sc Max-$c$-Lin-2}. A $t$-ordering is a mapping $\phi_t: V\rightarrow \{1,-1\}^t$. The vectors in $\{1,-1\}^{t}$ are totally ordered by lexicographic order, where $(+1) < (-1)$.\footnote{The reason for this ``funny'' ordering is that we want {\em true} to correspond to $-1$, and {\em false} to correspond to $1$, so that the parity function can be represented by multiplication, and the empty set of variables will correspond to {\em false} (i.e., even parity).} For example, we have $(+1,+1,-1)<(+1,-1,+1)$. As the vectors $\{1,-1\}^t$ are totally ordered, we interpret them as {\em positions} in $2^t$ buckets. Any $t$-ordering $\phi_t$ can be extended into a linear ordering $\phi : V \rightarrow \N$ which preserves the {\em bucket ordering} imposed by $\phi_t$, so that $\phi(v_i)<\phi(v_j)$ whenever $\phi_t(v_i)<\phi_t(v_j)$. We say $\phi$ is an extension of $\phi_t$ if it preserve the bucket ordering imposed by $\phi_t$.

The problem {\sc $t$-Ordering} is to find (or determine that there is no) $t$-ordering $\phi_t: V\rightarrow \{1,-1\}^t$ with weight $w_t(\phi_t,\mc{C})$ at least $\rho_t\cdot W + k$, where $W$ is the total weight of all constraints and $\rho_t$ is the expected fraction of weighted constraints satisfied by a uniform random assignment, i.e., $\rho_t\cdot W := \Ex[w_t(\phi_t,\mc{C})]$.

We can view an instance $(V,\mc{C},k)$ of {\sc 3-Linear Ordering} as an instance of {\sc $t$-Ordering} by introducing a new payoff function. For each constraint $e=(v_1, v_2, v_3)$, a payoff function $g_{e}:(\{1,-1\}^{t})^3 \rightarrow [0,1]$ maps $\tilde{x}=(\tilde{x}_1,\cdots , \tilde{x}_{3t})\in (\{1,-1\}^{t})^3$, a $t$-ordering $\phi_t$ restricted to the tuple $e$, to the probability that its random extension satisfies $e$. The weight of $e$ remains the same.  Note that $g_{e}(\tilde{x})$ is nonzero if and only if the three vectors $(\tilde{x}_1,\ldots, \tilde{x}_t), (\tilde{x}_{t+1},\ldots,\tilde{x}_{2t})$ and $(\tilde{x}_{2t+1},\ldots ,\tilde{x}_{3t})$ are lexicographically non-decreasing. By $w_t(\phi_t,\mc{C})$, we refer to the value $\sum_{e\in \mc{C}}w_e\cdot g_e(\phi_t)$.

As the function $g_e$ is defined on $x\in \{1,-1\}^{3t}$, it can be uniquely represented as a multilinear polynomial $$g_e(x)=\sum_{S\subseteq [3t]} \hat{g}_e(S)\cdot \chi_S(x).$$  Consequently, the weight of a $t$-ordering $\phi_t$ is expressed by the polynomial $g(x)=\sum_{e\in \mc{C}} w_e\cdot g_e(x)$ on $x\in \{1,-1\}^{tn}$. In other words, for every $t$-ordering $\phi_t$ we have $g(\tilde{x})=w_t(\phi_t,\mc{C})$. Here $(\tilde{x}_{(i-1)t+1},\ldots, \tilde{x}_{it}):=\phi_t(v_i)$ for each $i\in [n]$. The polynomial $g(x)$ can be rewritten as follows:

\begin{eqnarray}\label{eq:tord}
g(x) &=& \sum_{e\in \mc{C}} w_e \cdot g_e(x) = \sum_{e\in \mc{C}} w_e \cdot \sum_{S\subseteq [3t]} \hat{g}_e(S)\cdot \chi_S(x) \nonumber \\
&=& \sum_{S\subseteq [nt], |S|\leq 3t} \left( \sum_{e\in \mc{C}} w_e\hat{g}_e(S) \right)\chi_S(x) \nonumber \\
&=& \sum_{S\subseteq [nt], |S|\leq 3t} \hat{g}(S)\chi_S(x)
\end{eqnarray}

The following observation is folklore.

\begin{observation}
Let $(V,\mc{C},k)$ be an instance of {\sc 3-Linear Ordering}. Then $\rho = \rho_t$. In particular, if every constraint $e\in \mc{C}$ has arity exactly $3$, we have $\rho=\rho_t=\frac{1}{3!}$.
\end{observation}
\begin{proof}
Let $g_e(x)$ be the payoff function for constraint $e$ in the associated {\sc $t$-Ordering} instance. We have $\sum_{e'\in \Delta(e)} \Ex[g_{e'}(\phi_t)] = 1$ since we satisfy exactly one constraint from $\Delta(e)$ regardless of $\phi_t$. Moreover, we know $\Ex[g_{e'}(\phi_t)]$ takes the same value for all $e' \in \Delta(e)$. For $|\Delta(e)|=|e|!$, it follows that $\Ex[g_e(\phi_t)]=\frac{1}{|e|!}$. Observe $\rho_t\cdot W=\Ex[g(x)]=\sum_{e\in \mc{C}} w_e \Ex[g_{e}(x)] = \sum_{e\in \mc{C}} w_e\frac{1}{|e|!} = \rho\cdot W$. This completes the proof.
\end{proof}

By averaging argument, we observe the following.

\begin{observation}
Let $(V,\mc{C},k)$ be an instance of {\sc 3-Linear Ordering}, for any $t$-ordering $\phi$, we have $w_t(\phi_t,\mc{C}) \leq \max_{\phi} w(\phi,\mc{C})$, where $\phi$ is taken over all extensions of $\phi_t$.
\end{observation}

An easy consequence is:
\begin{observation}\label{yesyes}
Let $(V,\mc{C},k)$ is an instance of {\sc $3$-Linear Ordering}. If it has a $t$-ordering $\phi_t$ with $w_t(\phi_t,\mc{C})\geq \rho\cdot W + k$, then there is a linear ordering $\phi$ with $w(\phi,\mc{C})\geq \rho \cdot W + k$.
\end{observation}

\begin{observation}
For a constraint $e$ of arity $c$, the fourier coefficients of $g_e$ are $\frac{j}{c!2^{ct}}$ with integral $j \in [-c!2^{ct},c!2^{ct}]$. Accordingly, the same condition holds for the fourier coefficients of $g$.
\end{observation}
\begin{proof} Given two vectors $u$ and $v$, denote their concatenation $u \diamond v$ by the vector $w$ where $w_i = u_i$ for all $i=1,\ldots,|u|$, and $w_{|u|+i} = v_i$ for all $j=1,\ldots,|v|$.
It is not difficult to see that $g_e(x)=1/(n_1!\cdot n_2!\cdot \cdots n_{2^t}!)$ if \[x=(+1,+1,\cdots, +1)^{n_1}\diamond (+1,\cdots,+1,-1)^{n_2} \diamond \cdots \diamond (-1,-1,\cdots ,-1)^{n_{2^t}},\] and $g_e(x)$ takes 0 otherwise. Here $n_1+\cdots + n_{2^t}=c$. Hence, $g_e(x)=j/c!$ for some integer $j\in [0,c!]$. The claim follows from the the inverse fourier transform $\hat{g}_e(S)=\Ex_{x\in \{1,-1\}^{ct}} [g_e(x)\chi_S(x)]$ and the fact that the weight of a constraint $e$ is integral.
\end{proof}

Recall that every monomial in the polynomial $g(x)$ has degree at most $3t$. Therefore, the polynomial can be interpreted as an instance of {\sc Max-$3t$-Lin-2 Above Average}. Indeed, associate each monomial on the subset $S \subseteq [tn]$ in $g(x)$ with a linear equation $\sum_{i\in S}x_i=0$ of weight $3!2^{3t}\cdot \hat{g}(S)$ if $\hat{g}(S)$ is positive, and an equation $\sum_{i\in S}x_i=1$ of weight $-3!2^{3t}\cdot \hat{g}(S)$ if $\hat{g}(S)$ is negative. (Here we multiply by the factor of $3!2^{3t}$ to make the weights integral.) Let $F(\mc{C})$ be the instance of {\sc Max-$3t$-Lin-2 Above Average} constructed from an instance $\mc{C}$ of {\sc $t$-Ordering} with arity at most three. The next lemma claims that $F(\mc{C})$ is equivalent to the {\sc $t$-Ordering} instance. We omit the proof, which is essentially the same as Lemma~\ref{csp2lin}.

\begin{lemma}\label{tord2lin}
Let $(V,\mc{C},k)$ be an instance of {\sc $t$-Ordering}. There exists a $t$-ordering $\phi_t$ with $w(\phi_t)\geq \rho\cdot W + k/(c!2^{tc})$ if and only if the corresponding instance $F(\mc{C})$ of {\sc Max-$c$-Lin-2 Above Average} has an assignment of weight at least $W_{F(\mc{C})}/2 +k/2$, where $W_{F(\mc{C})}$ is the sum of weight of equations in $F(\mc{C})$.
\end{lemma}

Let $F(\mc{C})$ be an instance {\sc Max-$3t$-Lin-2 Above Average} obtained from {\sc $t$-Ordering} instance $(V,\mc{C},k)$. We can apply Theorem~\ref{maxclin2yes} from the previous section, and either decide that $F(\mc{C})$ is a yes-instance or that $F(\mc{C})$ has at most $(3t(3t+1)/2)k$ variables in polynomial time. By Lemma \ref{tord2lin}, the {\sc $t$-Ordering} instance is yes in the first case and thus, the {\sc 3-Linear Ordering} instance is yes as well by Observation \ref{yesyes}. The second outcome, however, does not necessarily imply a kernel for the {\sc 3-Linear Ordering} instance $(V,\mc{C},k)$.

To illustrate, let us consider {\sc Max Acyclic Subgraph}. Let us take $t=1$ and consider the associated {\sc $t$-Ordering} instance. For an edge constraint $e=(v_1, v_2)$, we have the payoff function $g_e(v_1,v_2)=1/2+v_1/4-v_2/4$ for $t$-ordering. Consider an instance graph $\overrightarrow{C_n}=v_1,\ldots , v_n$, in which the oriented edges of $\overrightarrow{C_n}$ form the constraints. Then $w_t(\phi_t,\overrightarrow{C_n})=n/2$ for any $\phi_t \in \{1,-1\}^n$. Hence the corresponding {\sc Max-$2$-Lin-2 Above Average} instance $F(\overrightarrow{C_n})$ is a no-instance and the number of variables is bounded (in fact, it has no variables). But there is a linear ordering $\phi$ with $w(\phi,\overrightarrow{C_n})=n-1$, and thus $\overrightarrow{C_n}$ is a yes-instance for any $k\leq (n-2)/2$. Moreover, the instance $\overrightarrow{C_n}$ is irreducible. Hence the gap between the number of variables in a no-instance $F(\overrightarrow{C_n})$ and $|V|$ in the {\sc 3-Linear Ordering} can be arbitrarily large.

In the example above, every variable in $V=\{v_1,\ldots , v_m\}$ for $\overrightarrow{C_n}$ disappears in the instance $F(\mc{C})$ when $t=1$. We show that $t=2$ is enough to ensure that every variable of $V$ is ``represented" in $F(\mc{C})$.

\begin{comment}
To resolve the issue, our approach is twofold: to reduce an instance $\mc{C}$ of {\sc $c$-Linear Ordering} into an equivalent one $\mc{C}'$ so that we get rid of inessential variables and constraints. Also we set $t$ large enough to $t$-Ordering is a good approximate of a linear ordering, but not too large to keep the instance size small. With appropriate reduction rules and $t$, we shall see that all variables of $\mc{C}'$ survive in its corresponding {\sc Max-$c$-Lin-2 Above Average} instance. Then by Corollary \ref{maxccspker} we have a problem kernel for {\sc $c$-Linear Ordering} with linear number of variables. In fact, it is shown that for arity $c=2$ and 3, fixing the constant $t=2$ is enough to ensure such property.
\end{comment}

Let $(V,\mc{C},k)$ be an instance of {\sc 3-Linear Ordering} and $t=2$. In the polynomial $g(x)$, we construe the variables $x_{2i-1}, x_{2i}$ as representing the position of the variable $v_i$ in a $t$-ordering $\phi_t$, i.e. $\phi_t(v_i)=(\tilde{x}_{2i-1}, \tilde{x}_{2i})$. We say that a variable $v_i\in V$ is {\em represented} in the associated {\sc $t$-Ordering} instance if $x_{2i-1}$ or $x_{2i}$ appears in the instance $F(\mc{C})$. We also say that a variable $v_i\in V$ is {\em represented} in a monomial $\hat{g}(S)\chi_S (x)$ if $S\cap \{2i-1,2i\}\neq \emptyset$ and $\hat{g}(S)\neq 0$. Obviously, $v_i$ is represented in the associated {\sc $t$-Ordering} instance if and only if there is a monomial in $g(x)$ representing $v_i$.

\begin{comment}
if and only if there is  which expresses the position of $v$ in a $t$-Ordering. Recall that some variables from $V$ may not appear after we simplify the function $g(x)$ into the form of \eqref{eq:tord}. We shall reduce an instance $(V,\mc{C},k)$ of {\sc $c$-Linear Ordering} into $(V',\mc{C}',k')$ so that in the reduced instance, every variable is represented in the associated {\sc $t$-Ordering} instance. The remaining issue is how to ensure that every variable of the reduced instance is represented in the associated {\sc $t$-Ordering} instance.
\end{comment}

An important observation is that a monomial of $g_e(x)$ which represents {\em every} variable of $var(e)$ for a constraint $e$ cannot be canceled out in $g(x)$ by a monomial generated by a constraint $e'$ with $|e|=|e'|$ and $var(e')\neq var(e)$. Therefore, we shall study the sum $\sum_{e'\in \mc{C}(e)} w_{e'}\cdot g_{e'}(x)$, and prove that there is a single monomial in the sum representing every variable of $var(e)$.

\begin{lemma}\label{represented2}
Let $I=(V,\mc{C},k)$ be an irreducible instance and $t=2$. For a constraint $e=(v_1,v_2)\in \mc{C}$, the fourier coefficient $\hat{g}_e(S)$ is nonzero for $S=\{1,3,4\}$.
\end{lemma}
\begin{proof} Given two vectors $u$ and $v$, denote their concatenation $u \diamond v$ by the vector $w$ where $w_i = u_i$ for all $i=1,\ldots,|u|$, and $w_{|u|+i} = v_i$ for all $j=1,\ldots,|v|$.
We recall that $g_e(x)$ is nonzero only if $(x_1,x_2)\leq (x_3,x_4)$. By the inverse fourier transform, we have $$\hat{g}_e(S)=\Ex_{x\in \{1,-1\}^4} [g_e(x)\chi_S(x)] = \frac{1}{2^4}\left(\sum_{x=\hat{x}_1\diamond \hat{x}_1}g_e(x)\chi_S(x) + \sum_{x=\hat{x}_1\diamond \hat{x}_2, \hat{x}_1 < \hat{x}_2}g_e(x)\chi_S(x) \right),$$ where $\hat{x}_i \in \{1,-1\}^2$. If $x$ is of the form $\hat{x}^2$ as in the first sum, we have $\chi_S(x)=x_1x_3x_4=x_4$ and the sum over such vectors equals zero. If $x=\hat{x}_1\diamond \hat{x}_2$ and $\hat{x}_1 < \hat{x}_2$, then it is tedious to check that the sum equals -2. The claim follows.
\end{proof}

\begin{lemma}\label{represented3}
Let $I=(V,\mc{C},k)$ be an irreducible instance and $t=2$. For a constraint $e=(v_1,v_2,v_3)\in \mc{C}$, the fourier coefficient $\hat{g}_e(S)$ is $\frac{-2}{2^6}$ for $S=\{1,3,5,6\}$ and $S=\{1,2,3,5\}$.
\end{lemma}
\begin{proof} Recall that we denote the concatenation of two vectors $u \diamond v$ by the vector $w$ where $w_i = u_i$ for all $i=1,\ldots,|u|$, and $w_{|u|+i} = v_i$ for all $j=1,\ldots,|v|$.

Consider $S=\{1,3,5,6\}$. The case when $S=\{1,2,3,5\}$ is symmetric. We have
\begin{eqnarray}
& &\hat{g}_e(S)\nonumber \\
 &=& \frac{1}{2^6}\left(\sum_{x=\hat{x}_1^3}g_e(x)\chi_S(x) + \sum_{x=\hat{x}_1^2\diamond \hat{x}_2}g_e(x)\chi_S(x) + \sum_{x=\hat{x}_1\diamond \hat{x}_2^2}g_e(x)\chi_S(x) + \sum_{x=\hat{x}_1\diamond \hat{x}_2\diamond \hat{x}_3}g_e(x)\chi_S(x)\right)\nonumber
\end{eqnarray}
\noindent where $\hat{x}_i \in \{1,-1\}^2$ and the sum is taken over all $x=\hat{x}_1\diamond \hat{x}_2 \diamond \hat{x}_3$ such that $\hat{x}_1\leq \hat{x}_2 \leq \hat{x}_3$. If $x=\hat{x}^3$ or of the form $\hat{x}_1^2\diamond \hat{x}_2$, we have $\chi_S(x)=x_1x_3x_5x_6=x_5x_6$ and the sum over such vectors equals zero. If $x=\hat{x}_1\diamond \hat{x}_2^2$, we have $\chi_S(x)=x_1x_6$ it is tedious to check that the sum equals zero. If $x$ is of the form $\hat{x}_1\diamond \hat{x}_2\diamond \hat{x}_3$, observe that the sum equals -2. The claim follows.
\end{proof}

\begin{lemma}\label{represented}
Let $I=(V,\mc{C},k)$ be an irreducible instance and $t=2$. Then every variable $v \in V$ is represented in the associated {\sc $t$-Ordering} instance.
\end{lemma}
\begin{proof}
Suppose $v\in V$ appears in some ternary constraint $e \in \mc{C}$. Consider the sum $\sum_{e'\in \mc{C}(e)} w_{e'}\cdot g_{e'}(x)$. By Lemma \ref{absentpair}, we know that there is an absent pair $e_1$, $e_2$ such that $var(e_1)=var(e_2)=var(e)$ and $e_1,e_2\notin \mc{C}(e)$. Let $e_1=(v_1,v_2,v_3)$. We take $S=\{1,3,4,5\}$. Then by Lemma \ref{represented3}, the fourier coefficient $\hat{g}_{e'}(S)$ is negative for every $e'\in \mc{C}(e)$. Observe that the monomial corresponding to $S$ in the sum $\sum_{e'\in \mc{C}(e)} g_{e'}(x)$ represents all variables in $var(e)$. Moreover, this monomial is not canceled out in $g(x)$. Therefore, we conclude that $v$ is represented by the associated {\sc $t$-Ordering} instance.

If $v\in V$ does not appear in any ternary constraint, then there is a binary constraint $e \in \mc{C}$ in which $v$ appears since the instance is reduced with respect to Redundancy Rule. Let $e=(v_1,v_2)$. As the instance is reduced with respect to Cancelation Rule, $(v_2,v_1)\notin \mc{C}$. Now Lemma \ref{represented2} implies that $v$ is represented by the associated {\sc $t$-Ordering} instance.
\end{proof}

\begin{proofof}{Theorem~\ref{linker3per}}
Let $I=(V,\mc{C},k)$ be an irreducible instance of {\sc 3-Linear Ordering} and assume $I$ is a no-instance. By Observation \ref{yesyes}, the associated {\sc $t$-Ordering} instance is no. By Lemmas~\ref{tord2lin} and Theorem~\ref{maxclin2}, the instance $F(\mc{C})$ of {\sc Max-$6$-Lin-2 Above Average} has less than $15k$ variables. Finally, Lemma \ref{represented} says that every variable in $V$ is represented in the associated {\sc $t$-Ordering instance} and thus $|V|\leq 15k$.
\end{proofof}

The problem kernel for {\sc 2-Linear Ordering} follows immediately.

\begin{corollary}\label{linker2per}
Let $I=(V,\mc{C},k)$ be an irreducible instance of {\sc 2-Linear Ordering}. If $I$ is a no-instance, we have $|V|< 10k$ variables.
\end{corollary}
\begin{proof}
The proof remains the same as in the proof of Theorem~\ref{linker3per} except that the instance $F(\mc{C})$ of {\sc Max-$4$-Lin-2 Above Average} has less than $10k$ variables.
\end{proof}

Combined with an exact $O^*(2^n)$-algorithm\footnote{The notation $O^*$ suppresses polylogarithmic factors}, our result in this section leads to the following result.

\begin{theorem}\label{singleexp}
The problem $\text{\sc Max-$c$-Permutation CSP Above Average}$ can be solved in time $2^{O(k)}$ for $c=2,3$.
\end{theorem}
\begin{proof}
Using Proposition \ref{focuslin}, $\text{\sc Max-$c$-Permutation CSP Above Average}$ instances can be transformed into an equivalent instance of {\sc $c$-Linear Ordering}. Due to the result stated as Theorem~\ref{linker3per} and Corollary~\ref{linker2per}, we either know the instance is a yes-instance or attain an equivalent instance with less than $15k$ variables for $c=3$ (or an instance with less than $10k$ variables for $c=2$). Hence, an $O^{\star}(2^n)$-algorithm to exactly compute the maximum number of satisfiable constraints on a $n$-variable instance of {\sc Exact 3-Linear Ordering} will yield a desired result. Here, an instance $I$ of {\sc Exact 3-Linear Ordering} is given as a pair $(V,\mc{C})$ and the task is to find a linear ordering on $V$ so as to {\em maximize} the number of satisfied constraints in $\mc{C}$. We give an exposition of such an algorithm for $c=3$. An analogous observation applies to $c=2$.

Bodlaender et. al \cite{BFKKT11} presents algorithms for {\sc Vertex Ordering} problems which runs in $O^{\star}(2^n)$ time and $O^{\star}(2^n)$ space, or $O^{\star}(4^n)$ time and polynomial space. Let $V$ be a set of elements, which may be vertices in graph problems or variables in our permutation CSP context. For a linear ordering $\pi$ on $V$, we denote the set $\{w\in V: \pi(w)<\pi (v)\}$ by $\pi_{<,v}$. Consider a function $f$ from the domain of triples $(G,S,v)$ to an integer, where $G$ is a graph, $S\subseteq V(G)$ and $v\in V(G)$.  It is shown in $\cite{BFKKT11}$ that if $f$ is polynomial time computable, the value $$\min_{\pi} \sum_{v\in V} f(G,\pi_{<,v},v),$$ where $\pi$ is taken over all possible linear orderings, can be computed either in $O^{\star}(2^n)$ time and $O^{\star}(2^n)$ space, or in $O^{\star}(4^n)$ time and polynomial space. Alerted readers might notice that the former uses dynamic programming in Held-Karp style and the latter employs recursion instead. We point out that although they describe the algorithms in the context of graph problems, the validity does not depend on whether the relations on $V$ are of arity two or not. It remains to formulate the {\sc Exact 3-Linear Ordering} problem to fit in the setting.

We take $f(I, S ,v):=|\{e=(a,v,c)\in \mc{C}:a\in S, c\in V\setminus (S\cup \{v\}\}|$ and note that $f$ can be computed in polynomial time given a triple $I=(V,\mc{C}), S\subseteq V$ and $v\in V$. To see that $\min_{\pi} \sum_{v\in V} -f(I,\pi_{<,v},v)=\max_{\pi}\sum_{v\in V} f(I,\pi_{<,v},v)$ equals the optimal value of the {\sc Exact 3-Linear Ordering} instance $I$, it suffices to observe the followings: given a linear ordering $\pi$, (a) the family $\{e=(a,v,c)\in \mc{C}\}, v\in V$ {\em partitions} the constraint set $\mc{C}$, (b) a constraint $e=(a,v,c) \in \mc{C}$ contributes one to $f(I,\pi_{<,v},v)$ if and only if $\pi$ satisfies $e$. Finally, we note that the extension of the formulation to weighted instances and instances with (some) constraints of arity two is straightforward.
\end{proof}

Closing this section, we point out that the recent work in \cite{guruswami2011}, independently of our paper, also explores the idea of ensuring a monomial which represents a variable in the multilinear polynomial. They consider the {\sc $c$-Linear ordering} problem in which every variable occurs in a bounded number of constraints and show that approximation beyond the random assignment threshold $1/c!$ is achievable. It is interesting to note as well that their motivation is to gain over the random assignment threshold.

\paragraph{\em Acknowledgement} The authors would like to thank Daniel Gon\c{c}alves for valuable discussion which inspired the results of Section~\ref{sec:per}.

\bibliographystyle{abbrv}
\bibliography{biblio}

\appendix

\section{Proof of Lemma~\ref{csp2lin}}\label{appendixproof}

%\begin{comment}
\begin{proofof}{Lemma~\ref{csp2lin}} Let $F = \{(f_1,S_1,w_1),\ldots, (f_m,S_m,w_m)\}$ be a given set of constraints and let $AVG$ be the expected weight in $F$ of a uniform random variable assignment. (For definitions, see the Preliminaries.) First, convert each constraint function $f_i : \{0,1\}^{|S_i|} \rightarrow \{0,1\}$ into the form $f'_i : \{-1,1\}^{|S_i|} \rightarrow \{0,1\}$, replacing all $0$'s with $1$'s and all $1$'s with $-1$'s. This can be easily done via the linear transformation $\ell(x) = 1-2x$. Let $F'$ be the new set of constraints.

Let $s_i = |S_i|$. A well-known fact is that for any function $f_i : \{-1,1\}^{s_i} \rightarrow \{0,1\}$ there is always a unique multivariate polynomial $p_i(x_1,\ldots,x_{s_i})$ such that for all assignments $\vec{a} \in \{-1,1\}^{s_i}$, $p_i(\vec{a}) = f_i(\vec{a})$. This polynomial has the form \[p_i(x_1,\ldots,x_{s_i}) = \sum_{T \subseteq [s']} \alpha^{(i)}_T \left(\prod_{i \in T} x_i\right),\] where every $\alpha^{(i)}_T = j/2^{s_i}$ for some integer $j \in [-2^{s_i},2^{s_i}]$. Define $q_i(x) = p_i(x) - \alpha^{(i)}_{\varnothing}$. That is, $q_i$ equals $p_i$ minus the constant coefficient of $p_i$.

We claim that for all assignments $\vec{a}=(a_1,\ldots,a_n) \in \{0,1\}^{n}$, $\vec{a}$ has weight at least $AVG+k$ in $F$ if and only if $\sum_{i=1}^m w_i\cdot q_i(\ell(a_1),\ldots,\ell(a_n))\geq k$. To see this, observe that the sum of all constant coefficients of $p_i$ multiplied by the weight $w_i$ is exactly $AVG$, since this sum is

\[
\begin{array}{lcl}
\sum_{i=1}^m w_i \cdot \alpha^{(i)}_{\varnothing}
&=& \sum_{i=1}^m  w_i\cdot \Ex_{x \in \{-1,1\}^n} \left[p_i(x) - q_i(x)\right] = AVG - \sum_{i=1}^m w_i\cdot \Ex_{x \in \{-1,1\}^n} \left[q_i(x)\right] \\ \\
&=& AVG - \sum_{i=1}^m w_i \cdot \sum_{T \neq \varnothing} \alpha^{(i)}_T \Ex_{x \in \{-1,1\}^n} \left[\prod_{i \in T} x_i\right] = AVG,
\end{array}
\]

 where the last equality follows because any nontrivial product of random variables over $\{-1,1\}$ has expectation zero. Hence a $\{0,1\}$-assignment with weight $w$ in $F$ translates directly to a $\{-1,1\}$-assignment that makes $\sum_{i} w_i\cdot q_i = w - AVG$.

Now we reduce the problem of finding a $\{-1,1\}$-assignment such that $\sum_i q_i \geq k/2^c$ to finding a $\{0,1\}$-assignment to a {\sc Max-$c$-Lin-2} instance with weight at least $W/2 + k/2$. Let $r(x_1,\ldots,x_n) = \sum_i q_i(x_1,\ldots,x_n)$.

Associate each monomial in $r(x)$ with a linear equation $e(T) = 0$ of weight $2^c\cdot \sum_{i=1}^m w_i\cdot \alpha^{(i)}_T$, if this quantity is positive. If this quantity is negative, then associate with the equation $e(T) = 1$ of weight $-2^c\cdot \sum_{i=1}^m w_i\cdot \alpha^{(i)}_T$. Here $e(T)=\sum_{i\in T} y_i$. Notice that we need the $2^c$ factor in the weights in order to make them integral. For any $\{-1,1\}$-assignment $x$, take $y_i=\ell^{-1}(x_i)$ and note that $x_i=(-1)^{y_i}$. Hence

\[2^c\cdot \sum_{i=1}^m w_i\cdot \alpha^{(i)}_T \left(\prod_{i \in T} x_i\right) = 2^c \cdot \sum_{i=1}^m w_i\cdot \alpha^{(i)}_T  (-1)^{e(T)}\]

equals the weight of the corresponding equation $e(T)$ if $y$ satisfies it, and equals the negative of the weight of $e(T)$ if $y$ falsifies it. Lastly it remains to observe that any $\{0,1\}$-assignment $y$ to a {\sc Max-$c$-Lin-2} instance has a weight at least $W/2+k/2$ if and only if the weight of satisfied equations minus the weight of falsified equations by $y$ is at least $k$. This shows that if there exists a $\{-1,1\}$-assignment such that $\sum_i w_i\cdot q_i \geq k/2^c$, there is a $\{0,1\}$-assignment to the corresponding {\sc Max-$c$-Lin-2} instance with weight at least $W/2 + k/2$. The proof of the opposite direction is straightforward.

In the above construction, {\sc Max-$c$-CSP} can be reduced to {\sc Max-$c$-Lin-2} in $O(2^c \cdot m)$ time (where $m$ is the number of constraints) and the number $m'$ of equations in the transformed instance will be $O(2^c \cdot m)$ in the worst case.\end{proofof}
%\end{comment}

\end{document}